\documentclass[letterpaper, 10 pt, conference]{ieeeconf}  %

\IEEEoverridecommandlockouts                              %
\usepackage{amsmath,amssymb,amsthm}
\usepackage[utf8]{inputenc}
\usepackage{xcolor}

\title{\LARGE \bf
Time-Discounted Incremental Input/Output-to-State Stability
}

\author{Sven Kn{\"u}fer and Matthias A. M{\"u}ller%
\thanks{\vspace{0.0cm}}%
\thanks{S. Kn{\"u}fer is with Robert Bosch GmbH,
        Driver Assistance, 70469 Stuttgart, Germany.
        {\tt\small knuefer@gmx.de}.}%
\thanks{M. A. M{\"u}ller is with Leibniz University Hannover,
        Institute of Automatic Control, 30167 Hannover, Germany.
        {\tt\small mueller@irt.uni-hannover.de}.}%
\thanks{This work was supported by the German Research Foundation under Grant MU3929-2/1, project number: 426459964.}%
}

\theoremstyle{plain}%
\newtheorem{thm}{Theorem}

\newtheorem{lem}[thm]{Lemma}

\newtheorem{prop}[thm]{Proposition}

\newtheorem{cor}[thm]{Corollary}

\theoremstyle{definition}
\newtheorem{defn}{Definition}

\theoremstyle{remark}
\newtheorem{rem}{Remark}

\renewenvironment{proof}{\begin{proofIEEE}}{\end{proofIEEE}}

\newcommand{\Xcon}{\mathbb{X}}
\newcommand{\Xfun}{\mathbb{X}}

\newcommand{\Yfun}{\mathbb{Y}}

\newcommand{\Ufun}{\mathbb{U}}

\newcommand{\Vcon}{\mathbb{V}}
\newcommand{\Vfun}{\mathbb{V}}

\newcommand{\Wcon}{\mathbb{W}}
\newcommand{\Wfun}{\mathbb{W}}

\newcommand{\met}{|}
\newcommand{\BL}{B_{L}}

\renewcommand{\gamma}{\beta_{w}}
\renewcommand{\delta}{\beta_{v}}
\renewcommand{\epsilon}{\beta_{u}}
\renewcommand{\varphi}{\beta_{y}}

\begin{document}

\maketitle
\thispagestyle{empty}
\pagestyle{empty}

\begin{abstract}

The present work provides two alternatives to formulate time-discounted incremental input/output-to-state stability (i-IOSS) as a suitable detectability notion for general nonlinear systems with non-additive disturbances.
Both formulations are related to existing i-IOSS notions which result as special cases.
Previous results that provide a sufficient Lyapunov-function condition for i-IOSS and that prove i-IOSS to be necessary for the existence of robustly stable full-order state observers are generalized to the presented time-discounted i-IOSS definition for general nonlinear systems.
For linear systems, explicit i-IOSS bounds are presented.

\end{abstract}

\section{INTRODUCTION}
\label{sec:introduction}

The concept of incremental input/output-to-state stability (i-IOSS) has become an effectively used nonlinear detectability notion in recent years.
Especially in the context of moving horizon estimation (MHE), it is successfully applied to achieve robust stability results~\cite{Allan_Rawlings_ACC2019,Hu_arXiv_2017,Hu_et_al_CDC15,Ji_et_al_MHE,knuefer2018,Muller_Aut_2017,Rawlings_Ji_JPC12}.
Generalizing the notion of incremental input-to-state stability (i-ISS)~\cite{Angeli_TAC02,AngeliSontagWang_TAC2000}, the concept of i-IOSS has been introduced in~\cite{Sontag_OSS} and shown to be necessary for the existence of full-order state observers.
While~\cite{Angeli_TAC02,AngeliSontagWang_TAC2000,Sontag_OSS} consider continuous-time systems with nonlinear process disturbances and additive output disturbances, the above mentioned MHE results apply i-IOSS for such systems in discrete-time.
In~\cite{Allan_Rawlings_ACC2019}, an according Lyapunov characterizations for i-IOSS in discrete-time is presented.
However, to the authors' best knowledge, no i-IOSS formulation for non-additive disturbances, i.e., for the most general nonlinear case has been formulated so far.
At the first glance, addressing non-additive disturbances might appear to be a notation-intensive but straight-forward generalization.
Though, our approach promises insight at least in the following two ways.
Firstly, when comparing two system trajectories, it allows to distinguish between different output trajectories and different measurement noise signals.
The same holds true when considering control inputs versus process disturbances.
Secondly, the special case of linear detectable systems reveals that the naturally resulting i-IOSS estimate separates inputs from process disturbances and outputs from output disturbances.
Hence, an according separation is also justified in the general nonlinear case.
An additional observation from the i-IOSS estimate for linear systems is time-discounting.
Intuitively, information content derived from previous inputs, disturbances, and outputs diminishes with proceeding time.
This intuition is only implicitly covered by the classical i-IOSS formulation which motivates time-discounted terms as observed in~\cite{Allan_Rawlings_TWCCC2018,Allan_Rawlings_ACC2019,knuefer2018}.
In this work, we consequently present a time-discounted i-IOSS formulation for general nonlinear systems with not necessarily additive disturbances, see Section~\ref{sec:nonlineardectectability}.
Moreover, we generalize existing Lyapunov techniques to give a sufficient condition for i-IOSS.
Section~\ref{sec:rgasobserver} furthermore shows that time-discounted i-IOSS is necessary for the existence of full-order state observers, which extends the results of~\cite{Sontag_OSS}.
Then, especially motivated by the common application in the field of MHE, we introduce an alternative sum-based i-IOSS formulation in Section~\ref{sec:summability} and relate it to the previous results.
Finally, Section~\ref{sec:lineariIOSS} addresses the special case of linear detectable systems.
We formally prove that detectability is equivalent to i-IOSS in the linear case and derive explicit i-IOSS estimates.
While this equivalence is a well known result, the authors are not aware of any discussion of such explicit i-IOSS bounds for linear systems in the literature.
Note that MHE results that provide guaranteed convergence rates without a-priori knowledge of the disturbances are based on i-IOSS estimates - even in the linear-quadratic case.
This motivates to address such explicit i-IOSS bounds even for linear systems.
Moreover, the linear case serves to evaluate conservativeness of the proof techniques and of the established estimates in the nonlinear case.

\section{SETUP}
\label{sec:setup}

Let
${\Xfun}$,
${\Ufun}$,
${\Yfun}$,
${0 \in \Wfun}$,
and
${0 \in \Vfun}$
be metric spaces with corresponding metrices ${\met{} \cdot, \cdot \met{}}$ and abbreviate ${\met{} \cdot, 0\met{}}$ by ${\met{} \cdot \met{}}$.
In the following, we consider nonlinear discrete-time system dynamics of the form
\begin{align}
	x(t+1) &= f( x(t), u(t), w(t) ),
  \label{eq:sys}
	\\
	y(t) &= h( x(t), u(t), v(t) ),
  \label{eq:out}
\end{align}
where ${t \in \mathbb{N} \ (\ni 0)}$ and where ${f: \Xfun \times \Ufun \times \Wfun \rightarrow \Xfun}$, ${h: \Xfun \times \Ufun \times \Vfun \rightarrow \Yfun}$ are some nonlinear functions constituting the nominal system dynamics and the output model, respectively.
In \eqref{eq:sys}-\eqref{eq:out}, ${u : \mathbb{N} \rightarrow \Ufun}$ gives the known control input and ${w : \mathbb{N} \rightarrow \Wcon}$ represents an a priori unknown process disturbance while ${v : \mathbb{N} \rightarrow \Vcon}$ defines an a priori unknown measurement noise.
An initial condition ${x_0 \in \Xcon}$, an input ${u}$, and a process disturbance ${w}$ lead to a state trajectory ${x: \mathbb{N} \rightarrow \Xcon}$ under \eqref{eq:sys}.
Finally, the measurement noise ${v}$ generates an output trajectory ${y: \mathbb{N} \rightarrow \Yfun}$ according to~\eqref{eq:out}.
Such a tuple ${ \{ x, u, w, v, y \} }$ satisfying \eqref{eq:sys}-\eqref{eq:out} for all ${t \in \mathbb{N}}$ is called a solution of system~\eqref{eq:sys}-\eqref{eq:out} in the following.

\begin{rem}
\label{rem:modeldiscussion}
Note that the above system formulation~\eqref{eq:sys}-\eqref{eq:out} aims to cover the most general system class of nonlinear inputs and outputs and nonlinear process and output disturbances.
As discussed in the introduction, this is in contrast to the existing formulations in the literature where additive disturbances, especially for the output model, are considered or either inputs or process disturbances are present exclusively.
Using additive output disturbances does leave degrees of freedom from a disturbance model perspective but constitutes a structural simplification.
In the context of existing i-IOSS results, the usage of additive output disturbances allows to represent differences of nominal outputs as disturbance differences.
For the sake of a clear separation between output signals and output disturbances (and according differences in their domains), the general nonlinear formulation is introduced in~\eqref{eq:out}.
From a controller point of view, the separation between process disturbances and inputs perfectly makes sense: the former is unknown and not manipulable while the latter is known and can be chosen.
To investigate detectability independent of the controller, the classic approach is either to take an input as given and include the input's influence directly into the system model eliminating the according function argument, or to take the input as unknown such that there is no use to distinguish between inputs and process disturbances.
In the present work, we however firstly aim for a formulation that separates the impact of both signal chains towards state and secondly the direct feed-through of the input ${u}$ towards the output ${y}$ constitutes a structural difference between the inputs and the process disturbances.
\end{rem}

In the context of nonlinear stability results, comparison functions according to the following definition are classically used.
For a collection of several results on such comparison functions see for instance~\cite{Kellett_MathConSigSys14}.

\begin{defn}[Comparison Functions]
\label{defn:comparisonfun}
A function ${\alpha: [0, \infty) \rightarrow [0, \infty)}$ is called ${\mathcal{K}}$-\emph{function}, i.e., ${\alpha \in \mathcal{K}}$, if ${\alpha}$ is continuous, strictly increasing, and ${\alpha(0)=0}$.
A function ${\alpha: \mathbb{N} \rightarrow [0, \infty)}$ is called ${\mathcal{L}}$-\emph{function}, i.e., ${\alpha \in \mathcal{L}}$, if ${\alpha}$ is non-increasing and ${\lim_{t\rightarrow\infty} \alpha(t)=0}$.
A function ${\beta: [0, \infty) \times \mathbb{N} \rightarrow [0, \infty)}$ is called ${\mathcal{KL}}$-\emph{function}, i.e., ${\beta \in \mathcal{KL}}$, if ${\beta(\cdot,t)} \in \mathcal{K}$ for each fixed ${t \in \mathbb{N}}$, and $\beta(r,\cdot) \in \mathcal{L}$ for each fixed ${r \in [0, \infty)}$.
\end{defn}

In the following, the abbreviation ${\kappa^{t}(r)}$ is used for the ${t}$-fold composition of any ${\kappa \in \mathcal{K}}$, i.e.,
${\kappa^{0}(r) := r}$ and ${\kappa^{t+1}(r) := \kappa(\kappa^{t}(r))}$
for ${t \in \mathbb{N}}$.

\section{NONLINEAR DETECTABLITY}
\label{sec:nonlineardectectability}

While the term \emph{detectability} is clearly defined for linear systems, many notions of detectability exist in the context of nonlinear systems.
Two main reasons might be that detectability of a certain state is in general not equivalent to detectability of arbitrary trajectories and that detectability is not only an issue of the output-to-state relation but also of the inputs' influence in the nonlinear context.
In previous works such as~\cite{Hu_arXiv_2017,Hu_et_al_CDC15,Ji_et_al_MHE,knuefer2018,Muller_Aut_2017,Rawlings_Ji_JPC12}, these observations are formulated in the detectability notion of incremental input/output-to-state stability (i-IOSS).
Adapted to the general system~\eqref{eq:sys}-\eqref{eq:out} with nonlinear disturbances and strengthened by explicit time-discounting, this work investigates the following notion of i-IOSS.

\begin{defn}[time-discounted i-IOSS]
\label{defn:DiscountingPredictoriIOSS}
System~\eqref{eq:sys}-\eqref{eq:out} is time-discounted incrementally input/output-to-state stable (i-IOSS) if there exist ${\beta, \gamma, \delta, \epsilon, \varphi \in \mathcal{KL}}$ such that, for any two solutions ${ \{ x, u, w, v, y \} }$ and ${ \{ \chi, \upsilon, \omega, \nu, \zeta \} }$ of~\eqref{eq:sys}-\eqref{eq:out}, the difference between the two trajectories remains bounded according to %
\begin{align}
\label{eq:defn:DiscountingPredictoriIOSS}
	|x(t), \chi(t)| &\leq \max \{ \beta( |x_{0}, \chi_{0}|, t ),
	\\
	& \qquad
  \notag
  \max _{1 \leq \tau \leq t} \{ \gamma( |w(t - \tau), \omega(t - \tau)|, \tau ),
	\\
	& \qquad \qquad \quad
  \notag
  \delta( |v(t - \tau), \nu(t - \tau)|, \tau ),
	\\
	& \qquad \qquad \quad
  \notag
  \epsilon( |u(t - \tau), \upsilon(t - \tau)|, \tau ),
	\\
	& \qquad \qquad \quad
  \notag
  \varphi( |y(t - \tau), \zeta(t - \tau)|, \tau ) \} \}
\end{align}
for all ${t \in \mathbb{N}}$%
.
\end{defn}

The above definition especially extends i-IOSS towards general nonlinear models with non-affine disturbances.
For classical i-IOSS, the difference between two arbitrary state trajectories is bounded in terms of the (i) their initial conditions, (ii) their inputs, (iii) their outputs.
In order to incorporate general disturbances, Definition~\ref{defn:DiscountingPredictoriIOSS} additionally introduces explicit terms for (iv) the process disturbances and (v) the output disturbances.
While for instance for output models of the form ${h(x, u, v) = \bar{h}(x) + v}$ differences of the additive output disturbances and differences of the outputs can be pulled together, the nonlinear setup of the present work requires to handle the influence of these differences separately in~\eqref{eq:defn:DiscountingPredictoriIOSS}.
An according statement applies for process disturbances that directly manipulate the input, i.e., ${f(x, u, w) = \bar{f}(x, u + w)}$.
While for the above two examples (${h(x, u, v) = \bar{h}(x) + v}$ and ${f(x, u, w) = \bar{f}(x, u + w)}$), the according bounding terms will be identical, i.e., ${\gamma = \epsilon}$ and ${\delta = \varphi}$, the terms in~\eqref{eq:defn:DiscountingPredictoriIOSS} in general allow to investigate the disturbances' influence independent of the inputs and outputs.
Especially the below comparison with i-IOSS results for linear systems in Section~\ref{sec:lineariIOSS} reveals that an i-IOSS estimate that explicitly depends on all five terms (i)-(v) might actually be the naturally expected from.

As classical, non-time-discounted i-IOSS only provides bounds for the disturbances' influences with respect to the maximum norm over time, asymptotic convergence of two trajectories can only be inferred indirectly via the decay rate of the initial error term, see, e.g.,~\cite{Hu_et_al_CDC15,Muller_Aut_2017}.
The special case of exponentially time-discounted i-IOSS has been considered in~\cite{knuefer2018}, and a suggestion to introduce explicit time-discounting has also been made in~\cite[Remarks~19 and~35]{Allan_Rawlings_TWCCC2018} and~\cite[Remark~6]{Allan_Rawlings_ACC2019}.
Definition 2 above gives a general, non-exponential version of time discounting, which also opens the way to Section~\ref{sec:summability}, in which a sum-based i-IOSS formulation is considered.

Concepts to show i-IOSS are investigated in~\cite{Allan_Rawlings_ACC2019,Angeli_TAC02,Bayer_et_al_ECC13,Ji_et_al_MHE,Sontag_OSS}.
These concepts can be extended to the above notion of time-discounted i-IOSS as stated in the following theorem.

\begin{thm}[i-IOSS Lyapunov Condition]
\label{thm:AllanRawlingsLyap}
Suppose there exist ${\mathcal{K}}$-functions ${\alpha_{1}, \alpha_{2}, \alpha_{3}, \rho_{w}, \rho_{v}, \rho_{u}, \rho_{y}}$ and a continuous function ${V: \mathbb{X} \times \mathbb{X} \rightarrow \mathbb{R}}$ such that
\begin{align}
  \label{eq:thm:AllanRawlingsLyap:Lyab}
  \alpha_{1}(| \bar{x}, \bar{\chi} |)
  &\leq
  V(\bar{x}, \bar{\chi})
  \leq
  \alpha_{2}(| \bar{x}, \bar{\chi} |)
\end{align}
is satisfied for all ${\bar{x}, \bar{\chi} \in \mathbb{X}}$ and such that
\begin{align}
  \label{eq:thm:AllanRawlingsLyap:Decrease}
  V(f(\bar{x}, \bar{u}, \bar{w}), f(\bar{\chi}, \bar{\upsilon}, \bar{\omega}))
  &\leq
  V(\bar{x}, \bar{\chi})
  - \alpha_{3}(V(\bar{x}, \bar{\chi})) %
  \\
  & \hspace{-2cm}
  \notag
  + \rho_{w}(| \bar{w}, \bar{\omega} |)
  + \rho_{v}(| \bar{v}, \bar{\nu} |)
  + \rho_{u}(| \bar{u}, \bar{\upsilon} |)
  \\
  & \hspace{-2cm}
  \notag
  + \rho_{y}(| h(\bar{x}, \bar{u}, \bar{v}), h(\bar{\chi}, \bar{\upsilon}, \bar{\nu}) |)
\end{align}
holds for all ${\bar{x}, \bar{\chi} \in \mathbb{X}}$, ${\bar{u}, \bar{\upsilon} \in \mathbb{U}}$, ${\bar{w}, \bar{\omega} \in \mathbb{W}}$, ${\bar{v}, \bar{\nu} \in \mathbb{V}}$.
Then the system~\eqref{eq:sys}-\eqref{eq:out} is time-discounted i-IOSS according to Definition~\ref{defn:DiscountingPredictoriIOSS}.
\end{thm}

\begin{proof}
The proof is a straight-forward generalization of the proof of Proposition~5 in~\cite{Allan_Rawlings_ACC2019}.
Using the construction in~\cite[Theorem~B.15]{Rawlings_Mayne_Diehl_MPC17}, we define
\begin{align}
  \label{eq:thm:AllanRawlingsLyap:construction}
  \kappa(r) := \frac{1}{2} r + \frac{1}{2} \max_{r' \in [0, r]} \{ r' - \alpha_{3}(r') \} %
\end{align}
such that ${\kappa \in \mathcal{K}}$ satisfies ${r > \kappa(r) > r - \alpha_{3}(r)}$ for all ${r \in (0, \infty)}$. %
By standard arguments it is shown that~\eqref{eq:thm:AllanRawlingsLyap:Decrease} implies
\begin{align}
  \label{eq:thm:AllanRawlingsLyap:contraction}
  V(f(\bar{x}, \bar{u}, \bar{w}), f(\bar{\chi}, \bar{\upsilon}, \bar{\omega}))
  &\leq
  \max \{
  \kappa(V(\bar{x}, \bar{\chi})),
  \\
  &
  \hspace{-2cm}
  \notag
  \phi_{w}(| \bar{w}, \bar{\omega} |),
  \phi_{v}(| \bar{v}, \bar{\nu} |),
  \phi_{u}(| \bar{u}, \bar{\upsilon} |),
  \\
  &
  \hspace{-2cm}
  \notag
  \phi_{y}(| h(\bar{x}, \bar{u}, \bar{v}), h(\bar{\chi}, \bar{\upsilon}, \bar{\nu}) |)
  \}
\end{align}
with ${\phi_{n}(r) := 4 \alpha_2(\alpha_3^{-1}(8 \rho_{n}(r))) + 4 \rho_{n}(r)}$ for ${n \in \{w, v, u, y\}}$.
This contraction leads to the required estimate~\eqref{eq:defn:DiscountingPredictoriIOSS}, cf. (8) in~\cite{Allan_Rawlings_ACC2019}, with ${\beta(r, t) := \alpha_{1}^{-1} \circ \kappa^{t} \circ \alpha_{2}}$ and ${\beta_{n}(\cdot, t) := \alpha_{1}^{-1} \circ \kappa^{t} \circ \phi_{n}}$. %
\end{proof}

In~\cite{Allan_Rawlings_TWCCC2018}, ${V}$ of Theorem~\ref{thm:AllanRawlingsLyap} is called an i-IOSS Lyapunov function.
Note that ${V}$ directly takes two arguments, i.e., it measures the distance between two states.
The decrease function ${\alpha_{3}}$ guarantees a distinct decrease of ${V}$ provided that the input, output, and disturbance differences of the two compared trajectories are small.
Due to the structure of the Lyapunov condition~\eqref{eq:thm:AllanRawlingsLyap:Decrease}, the decrease function ${\alpha_{3}}$ defines a common decrease rate for all terms in the desired estimate~\eqref{eq:defn:DiscountingPredictoriIOSS}.
Note that in general the different terms might have different decrease rates as an alternative proof technique for the special case of linear detectable systems reveals in Corollary~\ref{cor:LINiIOSSDirectSum} of Section~\ref{sec:lineariIOSS} below.

\section{RGAS OBSERVER}
\label{sec:rgasobserver}

This section investigates to which extent the time-discounted i-IOSS formulation of Definition~\ref{defn:DiscountingPredictoriIOSS} allows to preserve the classical results which relate existence of full-order state observers and the i-IOSS condition and which are formulated in~\cite{Sontag_OSS} for continuous-time systems with additive disturbances in a non-time-discounted way.
The following definition is an according generalization of~\cite[Definition~20]{Sontag_OSS}.

\begin{defn}[RGAS Observer]
\label{defn:DiscountingPredictorObserver}
A robustly globally asymptotically stable (full-order state) observer for system~\eqref{eq:sys}-\eqref{eq:out} is a system defined by
\begin{align}
\label{eq:defn:DiscountingPredictorObserver}
  \tilde{x}(t+1) = g(\tilde{x}(t), \tilde{u}(t), \tilde{w}(t), \tilde{v}(t), \tilde{y}(t))
\end{align}
with ${g : \mathbb{X} \times \mathbb{U} \times \mathbb{W} \times \mathbb{V} \times \mathbb{Y} \rightarrow \mathbb{X}}$ and ${\tilde{x} : \mathbb{N} \rightarrow \mathbb{X}}$ such that there exist ${\beta, \gamma, \delta, \epsilon, \varphi \in \mathcal{KL}}$ satisfying
\begin{align}
\label{eq:defn:DiscountingPredictorObserveriIOSS}
	|x(t), \tilde{x}(t)| &\leq \max \{ \beta( |x_{0}, \tilde{x}_{0}|, t ),
	\\
	& \qquad
  \notag
  \max _{1 \leq \tau \leq t} \{ \gamma( |w(t - \tau), \tilde{w}(t - \tau)|, \tau ),
	\\
	& \qquad \qquad \quad
  \notag
  \delta( |v(t - \tau), \tilde{v}(t - \tau)|, \tau ),
	\\
	& \qquad \qquad \quad
  \notag
  \epsilon( |u(t - \tau), \tilde{u}(t - \tau)|, \tau ),
	\\
	& \qquad \qquad \quad
  \notag
  \varphi( |y(t - \tau), \tilde{y}(t - \tau)|, \tau ) \} \}
\end{align}
for all ${t \in \mathbb{N}}$, %
all solutions ${ \{ x, u, w, v, y \} }$ of~\eqref{eq:sys}-\eqref{eq:out} and all solutions ${ \{ \tilde{x}, \tilde{u}, \tilde{w}, \tilde{v}, \tilde{y} \} }$ of~\eqref{eq:defn:DiscountingPredictorObserver}.
\end{defn}

Note that ${y}$ is the output of~\eqref{eq:sys}-\eqref{eq:out} while ${\tilde{y}}$ is an input of~\eqref{eq:defn:DiscountingPredictorObserver}.
Precisely, ${\tilde{y} \neq y}$ covers the case in which the disturbance of the output model~\eqref{eq:out} does not properly represent the actual disturbance affecting the channel \emph{output of the to-be-observed system towards observer input}.
Usually, one expects that the observer input ${\tilde{y}}$ equals ${h(x, u, v)}$, which might motivate to require~\eqref{eq:defn:DiscountingPredictorObserveriIOSS} only for solutions ${ \{ x, u, w, v, y \} }$ and ${ \{ \tilde{x}, \tilde{u}, \tilde{w}, \tilde{v}, \tilde{y} \} }$ that satisfy such a coupling condition.
However, relaxing Definition~\ref{defn:DiscountingPredictorObserver} in the described way leaves us with \emph{no statement at all} in arbitrary small neighborhoods of the coupling condition ${\tilde{y} = h(x, u, v)}$.
While for additive output disturbances ${v}$ as considered in~\cite{Allan_Rawlings_ACC2019,Sontag_OSS}, robustness against violations of the constraint ${\tilde{y} = h(x, u, v)}$ is implicitly represented by the output disturbance gain, the context of general nonlinear output disturbances requires to consider the case ${\tilde{y} \neq y}$ explicitly.
Accordingly, the case ${\tilde{u} \neq u}$ covers the neighborhood of the expected equivalence condition of the system's input ${u}$ and the observer's input ${\tilde{u}}$, i.e., of a potentially imprecise process disturbance model~\eqref{eq:sys}.
Finally, the inputs ${\tilde{w}}$ and ${\tilde{v}}$ allow to incorporate a priori guesses of the process disturbance and the measurement noise.
These observer inputs could for instance represent disturbance or parameter estimates that are gained by an additional external estimator.
However, classical implementations of~\eqref{eq:defn:DiscountingPredictorObserver} usually choose constant inputs ${\tilde{w}}$ and ${\tilde{v}}$ that represent nominal values.
While it appears to be a strong requirement to expect stability with respect to arbitrary a priori guesses, this is in fact crucial to show that the existence of an RGAS observer according to Definition~\ref{defn:DiscountingPredictorObserver} also implies the time-discounted i-IOSS property, see Proposition~\ref{prop:Observer2iIOSS} at the end of this section.

For the classical case of accurate disturbance models and zero a priori guesses, i.e., ${u = \tilde{u}}$, ${y = \tilde{y}}$, ${\tilde{v} = 0}$, and ${\tilde{w} = 0}$, a straight-forward consequence of~\eqref{eq:defn:DiscountingPredictorObserveriIOSS} is input-to-state stability of the observer-error with respect to the disturbances ${w}$ and ${v}$ in the sense of
\begin{align}
\label{eq:rgas:lyap}
	|x(t), \tilde{x}(t)| &\leq \max \{ \beta( |x_{0}, \tilde{x}_{0}|, t ),
  \\
	& %
  \max _{1 \leq \tau \leq t} \{ \gamma( |w(t - \tau)|, 0 ),
  \notag
  \delta( |v(t - \tau)|, 0 ) \} \}.
\end{align}
Moreover, the estimation error caused by a faulty initial estimation or by specific disturbances ${w(\tau)}$ or ${v(\tau)}$ decays asymptotically in~\eqref{eq:defn:DiscountingPredictorObserveriIOSS}.
Hence, the even stronger implication
\begin{align}
\label{eq:rgas:asymptotic}
  \lim_{t \rightarrow \infty} |w(t)| = 0
  \wedge
  \lim_{t \rightarrow \infty} |v(t)| &= 0
  \\
  \Rightarrow
  \qquad \qquad
  \lim_{t \rightarrow \infty} |x(t), \tilde{x}(t)| &= 0
\end{align}
results for the above sketched classical case of accurate disturbance models and zero a priori guesses.

Suppose system~\eqref{eq:sys}-\eqref{eq:out} contains asymptotically converging unobservable modes, then ${\beta}$ must bound the decrease rate of the slowest of such modes from above.
This idea illustrates that estimators which are based on the time-discounted i-IOSS condition do in general not allow to correct initial estimation errors any faster than the convergence rate of the slowest unobservable mode.
In the same lines, process disturbance and output noise introduce estimation errors.
As they might only effect parts of the state their corresponding decrease rates might be faster than the slowest unobservable mode while slower decrease rates are not to be expected.

The following two results are extensions of Lemma~21 respectively Proposition~23 in~\cite{Sontag_OSS}.

\begin{lem}[Output Injection Form]
\label{lem:OutputInjection}
Any RGAS observer according to Definition~\ref{defn:DiscountingPredictorObserver} must have the \emph{output injection form}, i.e., satisfy the identity
\begin{align}
\label{eq:lem:OutputInjection:identity}
  f(x_0, u_0, w_0) = g(x_0, u_0, w_0, v_0, h(x_0, u_0, v_0))
\end{align}
for all ${x_0 \in \mathbb{X}}$, ${u_0 \in \mathbb{U}}$, ${w_0 \in \mathbb{W}}$, ${v_0 \in \mathbb{V}}$.
\end{lem}

\begin{proof}
Consider arbitrary ${x_0 \in \mathbb{X}}$, ${u_0 \in \mathbb{U}}$, ${w_0 \in \mathbb{W}}$, ${v_0 \in \mathbb{V}}$ and choose
${\tilde{x}_0 = x_0}$, ${\tilde{u}_0 = u_0}$, ${\tilde{w}_0 = w_0}$, ${\tilde{v}_0 = v_0}$, and ${\tilde{y}_0 = h(x_0, u_0, v_0)}$.
Then Definition~\ref{defn:DiscountingPredictorObserver} requires
\begin{align}
	|x(1), \tilde{x}(1)| &\leq \max \{ \beta( |x_{0}, \tilde{x}_{0}|, 1 ),
	\\
	& \qquad
  \notag
  \max \{ \gamma( |w_0, \tilde{w}_0|, 1 ),
  \delta( |v_0, \tilde{v}_0|, 1 ),
	\\
	& \qquad \qquad \ \,
  \notag
  \varphi( |h(x_0, u_0, v_0), h(x_0, u_0, v_0)|, 1 ),
	\\
	& \qquad \qquad \ \,
  \notag
  \epsilon( |u_0, \tilde{u}_0|, 1 ) \} \}
\end{align}
via~\eqref{eq:defn:DiscountingPredictorObserveriIOSS} for ${t = 1}$.
Hence, ${|x(1), \tilde{x}(1)| = 0}$ or equivalently ${x(1) = \tilde{x}(1)}$ holds and consequently we obtain~\eqref{eq:lem:OutputInjection:identity}.
\end{proof}

\begin{prop}
\label{prop:Observer2iIOSS}
If an observer according to Definition~\ref{defn:DiscountingPredictorObserver} exists for system~\eqref{eq:sys}-\eqref{eq:out}, the system is time-discounted i-IOSS according to Definition~\ref{defn:DiscountingPredictoriIOSS}.
\end{prop}

\begin{proof}
Consider arbitrary ${x_0, \tilde{x}_0 \in \mathbb{X}}$, ${u, \tilde{u} : \mathbb{N} \rightarrow \mathbb{U}}$, ${w, \tilde{w} : \mathbb{N} \rightarrow \mathbb{W}}$, ${v, \tilde{v} : \mathbb{N} \rightarrow \mathbb{V}}$ resulting in a state trajectory ${x : \mathbb{N} \rightarrow \mathbb{X}}$ according to system~\eqref{eq:sys}-\eqref{eq:out}.
Apply the feedback ${\tilde{y}(t) := h(\tilde{x}(t), \tilde{u}(t), \tilde{v}(t))}$ for all ${t \in \mathbb{N}}$ to the observer dynamics~\eqref{eq:defn:DiscountingPredictorObserver}.
Consequently, the identity~\eqref{eq:lem:OutputInjection:identity} applies such that the observer state follows the dynamics
\begin{align}
\label{eq:prop:Observer2iIOSS:observer}
  \tilde{x}(t+1) = f(\tilde{x}(t), \tilde{u}(t), \tilde{w}(t))
\end{align}
for all ${t \in \mathbb{N}}$.
Hence, \eqref{eq:defn:DiscountingPredictorObserveriIOSS} directly gives the desired estimate~\eqref{eq:defn:DiscountingPredictoriIOSS}.
\end{proof}

The above two results, Lemma~\ref{lem:OutputInjection} and Proposition~\ref{prop:Observer2iIOSS} illustrate that the time-discounted i-IOSS formulation according to Definition~\ref{defn:DiscountingPredictoriIOSS} gives a natural generalization towards general non-linear systems, i.e., especially with non-additive output disturbances.
While the generalization of the i-IOSS estimate is rather straight-forward, Proposition~\ref{prop:Observer2iIOSS} shows that the classical condition for an RGAS observer needs to be extended towards robustness against arbitrary a-priori disturbance guesses ${\tilde{w}}$ and ${\tilde{v}}$, see Definition~\ref{defn:DiscountingPredictorObserver} and its discussion, in order to preserve the existing result that i-IOSS is necessary for the existence of a full-order state observer for general nonlinear systems.

\begin{rem}
Note that also~\cite[Proposition~6.1]{Angeli_TAC02} can be generalized in a similar way such that the following statement applies as well:
If the dynamics~\eqref{eq:defn:DiscountingPredictorObserver} are time-discounted incrementally input-to-state stable (i-ISS) with respect to all four inputs and if they satisfy the output-injection form~\eqref{eq:lem:OutputInjection:identity}, then they define an RGAS observer for~\eqref{eq:sys}-\eqref{eq:out} according to Definition~\ref{defn:DiscountingPredictorObserver}.
\end{rem}

\section{SUM-BASED i-IOSS FORMULATION}
\label{sec:summability}

This section investigates a sufficient condition that allows to replace the max-terms in~\eqref{eq:defn:DiscountingPredictoriIOSS} of Definition~\ref{defn:DiscountingPredictoriIOSS} by sums.
Our main motivation is that such a formulation naturally results for linear systems, see Remark~\ref{rem:thm:LINiIOSSLyap:sum} and Corollary~\ref{cor:LINiIOSSDirectSum} in Section~\ref{sec:lineariIOSS} below.
Hence, the question arises under which condition also nonlinear systems satisfy an according i-IOSS estimate.
Many MHE-results such as~\cite{Hu_arXiv_2017,Hu_et_al_CDC15,Ji_et_al_MHE,knuefer2018,Muller_Aut_2017,Rawlings_Ji_JPC12} that make use of the i-IOSS condition put the i-IOSS estimate and the utilized MHE cost function into relation in order to derive RGAS guarantees for the constructed estimators.
As classical MHE cost functions sum up cost terms over certain horizons, this gives an additional motivation to look for a sum-based i-IOSS formulation.

In order to obtain a well-defined sum-based upper bound, the utilized ${\mathcal{KL}}$-functions need to be summable according to the following definition, for which a sufficient condition is introduced in the proposition below.

\begin{defn}
\label{defn:summable}
A ${\mathcal{KL}}$-function ${\beta}$ is called \emph{summable} if there exists a bounding ${\mathcal{K}}$-function $\sigma$ such that ${\sum _{\tau = 0} ^{\infty} \beta( r, \tau ) \leq \sigma(r)}$ holds for all ${r \in \mathbb{R}}$.
\end{defn}

\begin{prop}
\label{prop:inducesummable}
Consider a ${\mathcal{K}}$-function ${\alpha}$.
If there exist ${\bar{r} \in (0, \infty)}$ and ${K \in (0, 1)}$ such that
\begin{align}
\label{eq:prop:inducesummable:MinIncrease}
  \alpha(r) \geq K r
\end{align}
holds for all ${r \in [0, \bar{r}]}$, then there exists a ${\mathcal{K}}$-function ${\kappa}$ that satisfies
\begin{align}
\label{eq:prop:inducesummable:KappaLowerBound}
  \kappa(r) \geq r - \alpha(r)
\end{align}
for all ${r \in [0, \infty)}$ and such that ${\beta(r, t) := \kappa^{t}(r)}$ is a summable ${\mathcal{KL}}$-function.
\end{prop}

\begin{proof}
Defining ${\kappa}$ according to~\eqref{eq:thm:AllanRawlingsLyap:construction} with ${\alpha_{3}}$ replaced by ${\alpha}$, \eqref{eq:prop:inducesummable:KappaLowerBound} is satisfied. 
Due to~\eqref{eq:prop:inducesummable:MinIncrease}, we furthermore obtain
\begin{align}
  \kappa(r)
  \label{eq:prop:inducesummable:kappaboundsmall}
  &\leq ( 1 - \frac{1}{2} K ) r
\end{align}
for all ${r \in [0, \bar{r}]}$ and
\begin{align}
  \kappa(r)
  \label{eq:prop:inducesummable:kappaboundlarge}
  &\leq r - \frac{1}{2} K \bar{r}
\end{align}
for all ${r \in [\bar{r}, \infty)}$.
In order to show boundedness of ${\sum _{\tau = 0} ^{\infty} \kappa^{\tau}( r )}$ for arbitrary ${r \in [0, \infty)}$, we split the sum into summands smaller ${\bar{r}}$ and summands larger ${\bar{r}}$.
For this purpose we observe that
\begin{align}
  \label{eq:prop:inducesummable:kappastarboundsmall}
  \kappa^{\tau}( r ) &\leq r - \frac{1}{2} K \bar{r} \tau
\end{align}
holds for all ${r \in [0, \infty)}$ and all ${\tau \in \mathbb{N}}$ with ${\tau \leq \bar{\tau}}$, ${\bar{\tau} := \max \{ 0, \left\lceil \frac{2(r - \bar{r})}{K \bar{r}} \right\rceil \} }$, due to~\eqref{eq:prop:inducesummable:kappaboundlarge}.
Moreover due to~\eqref{eq:prop:inducesummable:kappaboundsmall} we have
\begin{align}
  \label{eq:prop:inducesummable:kappastarboundlarge}
  \kappa^{\tau}( r ) &\leq ( 1 - \frac{1}{2} K )^{\tau - \bar{\tau}} \min \{ r, \bar{r} \}
\end{align}
for all ${r \in [0, \infty)}$ and all ${\tau \in \mathbb{N}}$ with ${\tau > \bar{\tau}}$.
All in all, we obtain
\begin{align}
  &\sum _{\tau = 0} ^{\infty} \kappa^{\tau}( r )
  \leq
  \label{eq:prop:inducesummable:sumsplit}
  \sum _{\tau = 0} ^{\bar{\tau}} \kappa^{\tau}( r )
  +
  \sum _{\tau = (\bar{\tau} + 1)} ^{\infty} \kappa^{\tau}( r )
  \\
  & \ \leq
  \sigma(r)
  :=
  \begin{cases}
  ( \frac{2}{K} + 1 ) r
  &
  \text{for} \
  r < \bar{r}
  \\
  \frac{1}{ K \bar{r}}
  ( r^2 + \bar{r}^2 )
  +
  \frac{1}{2}
  ( 3 \bar{r} - r )
  &
  \text{for} \
  r \geq \bar{r}
  \end{cases}
\end{align}
using~\eqref{eq:prop:inducesummable:kappastarboundsmall} to bound the finite sum and \eqref{eq:prop:inducesummable:kappastarboundlarge} to bound the last sum in~\eqref{eq:prop:inducesummable:sumsplit}, where ${\sigma}$ follows by a longer but straight-forward computation.
Finally, we observe that ${\sigma}$ is a ${\mathcal{K}}$-function, which concludes the proof.
\end{proof}

As the following theorem details, Theorem~\ref{thm:AllanRawlingsLyap} provides a sum-based i-IOSS formulation if the decrease function ${\alpha_{3}}$ satisfies the local linear lower-bound condition~\eqref{eq:prop:inducesummable:MinIncrease} introduced in Proposition~\ref{prop:inducesummable}.

\begin{thm}
\label{thm:iIOSSsummable}
If the conditions of Theorem~\ref{thm:AllanRawlingsLyap} are met and ${\alpha_{3}}$ satisfies~\eqref{eq:prop:inducesummable:MinIncrease} (i.e. is locally linearly lower-bounded at the origin), then the system is time-discounted i-IOSS according to Definition~\ref{defn:DiscountingPredictoriIOSS} with~\eqref{eq:defn:DiscountingPredictoriIOSS} replaced by
\begin{align}
\label{eq:thm:iIOSSsummable}
	\alpha_{1}(|x(t), \chi(t)|) &\leq \beta^{\Sigma}( |x_{0}, \chi_{0}|, t )
	\\
	& \qquad
  \notag
  + \sum _{\tau = 1} ^{t} ( \gamma^{\Sigma}( |w(t - \tau), \omega(t - \tau)|, \tau )
	\\
	& \qquad \qquad
  \notag
  + \delta^{\Sigma}( |v(t - \tau), \nu(t - \tau)|, \tau )
	\\
	& \qquad \qquad
  \notag
  + \epsilon^{\Sigma}( |u(t - \tau), \upsilon(t - \tau)|, \tau )
	\\
	& \qquad \qquad
  \notag
  + \varphi^{\Sigma}( |y(t - \tau), \zeta(t - \tau)|, \tau ) )
\end{align}
for all ${t \in \mathbb{N}}$ with ${\beta^{\Sigma}, \gamma^{\Sigma}, \delta^{\Sigma}, \epsilon^{\Sigma}, \varphi^{\Sigma} \in \mathcal{KL}}$.
Moreover, all these ${\mathcal{KL}}$-functions are summable.
\end{thm}

\begin{proof}
Due to~Proposition~\ref{prop:inducesummable}, ${\beta^{\Sigma}(r, t) := \kappa^{t}(r)}$ with ${\kappa}$ according to~\eqref{eq:thm:AllanRawlingsLyap:construction} is summable.
So is any composition ${\beta^{\Sigma}(\rho(r), t)}$ with ${\rho \in \mathcal{K}}$ arbitrary, i.e., especially ${\beta^{\Sigma}(r, t) := \kappa^{t} \circ \alpha_{2}}$ and ${\beta^{\Sigma}_{n}(\cdot, t) := \kappa^{t} \circ \phi_{n}}$ for ${n \in \{w, v, u, y\}}$.
According to the proof of Theorem~\ref{thm:AllanRawlingsLyap}, the desired estimate~\eqref{eq:thm:iIOSSsummable} consequently results from~\eqref{eq:defn:DiscountingPredictoriIOSS} by applying ${\alpha_{1}}$ on both sides and by replacing all maximizations with summations.
\end{proof}

\begin{rem}
\label{rem:thm:iIOSSsummable:max2sum}
Following the above proof, estimate~\eqref{eq:thm:iIOSSsummable} constitutes a loosened form of estimate~\eqref{eq:defn:DiscountingPredictoriIOSS} as the max-terms are simply replaced by sums, i.e., the upper bound in general increases.
However, if the linear lower bound~\eqref{eq:prop:inducesummable:MinIncrease} holds even globally, the sum-based formulation in~\eqref{eq:thm:iIOSSsummable} turns out to be the more straight-forward and stricter estimate.
In this case, the proof of Theorem~\ref{thm:AllanRawlingsLyap} can make use of ${\kappa(r) := (1 - K)r}$, i.e., a linear contraction function.
Hence the detour via the max-estimate~\eqref{eq:thm:AllanRawlingsLyap:contraction} is no more needed as a direct induction allows to derive an estimate according to~\eqref{eq:thm:iIOSSsummable} with ${\beta^{\Sigma}(r, t) := (1 - K)^{t} \alpha_{2}(r)}$ and ${\beta^{\Sigma}_{n}(r, t) := (1 - K)^{t} \rho_{n}(r)}$ for ${n \in \{w, v, u, y\}}$.
Note that this direct induction allows to arrive at the i-IOSS estimate~\eqref{eq:thm:iIOSSsummable} without sacrificing parts of the decrease function ${\alpha_3}$ to gain the max-estimate~\eqref{eq:thm:AllanRawlingsLyap:contraction}.
Therefore, the decrease rates in~\eqref{eq:thm:iIOSSsummable} are faster while the gains for the disturbances, the inputs, and the outputs are smaller compared to the ones that result for~\eqref{eq:defn:DiscountingPredictoriIOSS} in Theorem~\ref{thm:AllanRawlingsLyap}.
\end{rem}

\begin{rem}
\label{rem:thm:iIOSSsummable:sum2max}
To complete the comparison between the two i-IOSS estimates~\eqref{eq:defn:DiscountingPredictoriIOSS} and~\eqref{eq:thm:iIOSSsummable}, the question arises how to transform the sum-formulation in~\eqref{eq:thm:iIOSSsummable} to the max-formulation in~\eqref{eq:defn:DiscountingPredictoriIOSS}.
Due to the time-discounted formulation in~\eqref{eq:defn:DiscountingPredictoriIOSS}, there is no general answer for arbitrary summable ${\mathcal{KL}}$-functions.
However, for exponentially decreasing terms as discussed in Remark~\ref{rem:thm:iIOSSsummable:max2sum}, the according max-estimate~\eqref{eq:defn:DiscountingPredictoriIOSS} results if the decrease rate is partially sacrificed as for example in
\begin{align}
  \hspace{-0.1cm}
  \label{eq:proof:exposum2max}
  \sum _{\tau = 1} ^{t} \eta^{\tau} \theta_{\tau}
  \leq
  \sum _{\tau = 1} ^{t} \eta^{\frac{\tau}{2}} \max_{1 \leq \tilde{\tau} \leq t} \eta^{\frac{\tilde{\tau}}{2}} \theta_{\tilde{\tau}}
  \leq
  \frac{ \eta^{\frac{1}{2}} }{ 1 - \eta^{\frac{1}{2}} } \max_{1 \leq \tau \leq t} \eta^{\frac{\tau}{2}} \theta_{\tau}
\end{align}
with ${\eta \in [0, 1)}$ and ${\theta_{\tau} \in \mathbb{R}}$.
In particular, the crucial step to derive~\eqref{eq:defn:DiscountingPredictoriIOSS} from~\eqref{eq:thm:iIOSSsummable} is to apply an argument as in~\eqref{eq:proof:exposum2max} to each of the four terms in the sum on the right hand side of~\eqref{eq:thm:iIOSSsummable}. 
Finally observe that the classical, i.e., not time-discounted, i-IOSS estimates such as (3) in \cite{Allan_Rawlings_ACC2019} simply result by utilizing the upper bounding ${K}$-functions according to Definition~\ref{defn:summable}.
\end{rem}

\begin{rem}
\label{rem:SontagWang:max2sum}
Without any difficulties, we see that all arguments of the proofs of Lemma~\ref{lem:OutputInjection} and Proposition~\ref{prop:Observer2iIOSS} also hold true in case all maximizations in Definitions~\ref{defn:DiscountingPredictoriIOSS} and~\ref{defn:DiscountingPredictorObserver} are replaced by sums.
Hence, the fundamental results of Section~\ref{sec:rgasobserver} also apply for the sum-based i-IOSS definition.
\end{rem}

\section{LINEAR SYSTEMS AND i-IOSS}
\label{sec:lineariIOSS}

Now consider linear systems over ${\Xfun = \mathbb{R}^{n_{x}}}$, i.e,
\begin{align}
  \label{eq:syslin}
	f(x, u, w) &= A x + B u + E w,
	\\
  \label{eq:outlin}
	h(x, u, v) &= C x + D u + F v,
\end{align}
with ${\Ufun = \mathbb{R}^{n_{u}}}$, ${\Yfun = \mathbb{R}^{n_{y}}}$, ${\Wfun = \mathbb{R}^{n_{w}}}$, ${\Vfun = \mathbb{R}^{n_{v}}}$, let ${A}$, ${B}$, ${C}$, ${D}$, ${E}$, ${F}$ be matrices of corresponding dimensions, and use the canonical metric ${\met{} \cdot, \cdot \met{} = \| \cdot - \cdot \|}$.
This allows to formulate the following non-surprising result.

\begin{thm}
\label{thm:LINiIOSSLyap}
A linear system is detectable if and only if it is time-discounted i-IOSS according to Definition~\ref{defn:DiscountingPredictoriIOSS}.
\end{thm}

\begin{proof}
If a linear system is time-discounted i-IOSS then ${u, v, w, y \equiv 0}$ implies ${x \rightarrow 0}$ for ${t \rightarrow \infty}$ which leads to detectability in the linear case.
For the opposite direction, detectability of a linear system~\eqref{eq:syslin}-\eqref{eq:outlin} guarantees that there exist ${P, Q \in \mathbb{R}^{n_{x} \times n_{x}}}$ positive definite and symmetric and ${L \in \mathbb{R}^{n_{x} \times n_{y}}}$ such that
\begin{align}
  \label{thm:LINiIOSSLyap:discretelyap}
  A_{L}^{\top} P A_{L} = P - Q
\end{align}
with ${A_{L} = A + LC}$.
Consider arbitrary ${\bar{x}, \bar{\chi} \in \mathbb{R}^{n_{x}}}$, ${\bar{u}, \bar{\upsilon} \in \mathbb{R}^{n_{u}}}$, ${\bar{w}, \bar{\omega} \in \mathbb{R}^{n_{w}}}$, ${\bar{v}, \bar{\nu} \in \mathbb{R}^{n_{v}}}$ and define ${x_{\Delta} := \bar{x} - \bar{\chi}}$, ${x_{\Delta}^{+} := f(\bar{x}, \bar{u}, \bar{w}) - f(\bar{\chi}, \bar{\upsilon}, \bar{\omega})}$, ${u_{\Delta} := \bar{u} - \bar{\upsilon}}$, ${w_{\Delta} := \bar{w} - \bar{\omega}}$, ${v_{\Delta} := \bar{v} - \bar{\nu}}$, ${y_{\Delta} := h(\bar{x}, \bar{u}, \bar{v}) - h(\bar{\chi}, \bar{\upsilon}, \bar{\nu})}$ and ${\BL := B + L D}$.
Then we obtain
\begin{align}
  \label{thm:LINiIOSSLyap:outputinjection}
	x_{\Delta}^{+} &= A x_{\Delta} + B u_{\Delta} + E w_{\Delta}
  + L (y_{\Delta} - y_{\Delta})
  \\
  &= A_{L} x_{\Delta} + E w_{\Delta}
  \label{eq:thm:LINiIOSSLyap:differenceinteration}
  + L F v_{\Delta} + \BL u_{\Delta} - L y_{\Delta}.
\end{align}
Applying ${\| \cdot \|_{P} := \| P^{\frac{1}{2}} \cdot \|}$ to both sides and using the triangle-inequality allows to further derive
\begin{align}
	\| x_{\Delta}^{+} \|_{P}
  &=
  \| ( A_{L} x_{\Delta} + E w_{\Delta}
  \\
  & \qquad
  \notag
  + L F v_{\Delta}
  + \BL u_{\Delta} - L y_{\Delta} ) \|_{P}
  \\
  & \leq
  \label{eq:thm:LINiIOSSLyap:contraction}
  \| A_{L} \|_{P} \| x_{\Delta} \|_{P} + \| E w_{\Delta} \|_{P}
  \\
  & \qquad
  \notag
  + \| L F v_{\Delta} \|_{P}
  + \| \BL u_{\Delta} \|_{P} + \| L y_{\Delta} \|_{P}
  \\
  &\leq
  \label{eq:thm:LINiIOSSLyap:Lyap}
  \| x_{\Delta} \|_{P} - \alpha_{3} (\| x_{\Delta} \|_{P}) + \rho_{w} (\| w_{\Delta} \|)
  \\
  & \qquad
  \notag
  + \rho_{v} (\| v_{\Delta} \|)
  + \rho_{u} (\| u_{\Delta} \|) + \rho_{y} (\| y_{\Delta} \|)
\end{align}
with
\begin{align}
  \alpha_{3}(r) &:= ( 1 - \| A_{L} \|_{P} ) \, r
  \label{eq:thm:LINiIOSSLyap:decreasefunction}
  \\
  \label{eq:thm:LINiIOSSLyap:Wgain}
  \rho_{w}(r) &:= \sqrt{ \lambda_{\text{max}}(E^\top P E) } \, r
  \\
  \label{eq:thm:LINiIOSSLyap:Vgain}
  \rho_{v}(r) &:= \sqrt{\lambda_{\text{max}}((L F)^\top P (L F)) } \, r
  \\
  \label{eq:thm:LINiIOSSLyap:Ugain}
  \rho_{u}(r) &:= \sqrt{\lambda_{\text{max}}(\BL^\top P \BL) } \, r
  \\
  \label{eq:thm:LINiIOSSLyap:Ygain}
  \rho_{y}(r) &:= \sqrt{\lambda_{\text{max}}(L^\top P L) } \, r.
\end{align}
Hence, all conditions of Theorem~\ref{thm:AllanRawlingsLyap} are met with ${V(x_1, x_2) := \| x_1 - x_2 \|_{P}}$ which leads to an i-IOSS estimate according to Definition~\ref{defn:DiscountingPredictoriIOSS}.
\end{proof}

\begin{rem}
\label{rem:thm:LINiIOSSLyap:decreaserate}
For the decrease function ${\alpha_3}$ in~\eqref{eq:thm:LINiIOSSLyap:decreasefunction}, we note that
\begin{align}
  \| A_{L} \|_{P} ^2
  &=
  \max _{x \neq 0}
  \frac{x^\top A_{L}^\top P A_{L} x}{x^\top P x}
  \\
  &=
  1 - \min _{x \neq 0} \frac{x^\top Q x}{x^\top P x}
  \\
  &=
  1 - \lambda_{\text{min}}(P^{-\frac{1}{2}} Q P^{-\frac{1}{2}})
\end{align}
due to~\eqref{thm:LINiIOSSLyap:discretelyap}.
Moreover, \eqref{thm:LINiIOSSLyap:discretelyap} implies ${0 \leq I - P^{-\frac{1}{2}} Q P^{-\frac{1}{2}}}$ such that ${\| A_{L} \|_{P} \in [0, 1)}$ holds.
\end{rem}

\begin{rem}
\label{rem:thm:LINiIOSSLyap:sum}
Note that linearity of ${\alpha_{3}}$ allows to directly obtain an i-IOSS estimate according to~\eqref{eq:thm:iIOSSsummable}, see Remark~\ref{rem:thm:iIOSSsummable:max2sum}.
In this case the exponential decrease rate of all ${\mathcal{KL}}$-functions is given by ${\| A_{L} \|_{P}}$.
\end{rem}

\begin{rem}
\label{rem:thm:LINiIOSSLyap:quadratic}
The Lyapunov function defined in the proof of Theorem~\ref{thm:LINiIOSSLyap} is given by the square root of the usually expected quadratic term and hence lacks differentiability at the origin.
Squaring both sides of estimate~\eqref{eq:thm:LINiIOSSLyap:contraction} and upper bounding the resulting cross-terms under sacrificing an arbitrary part of the decrease rate, alternatively allows to derive a quadratic Lyapunov function $V$ at the cost of larger gains ${\rho}$.
\end{rem}

While the fact that detectability and time-discounted i-IOSS are equivalent for linear systems is rather expected, see remarks in~\cite{Angeli_TAC02,Ji_et_al_MHE,Rawlings_Mayne_Diehl_MPC17} for the non-time-discounted case, the proof of Theorem~\ref{thm:LINiIOSSLyap} allows to gain insight into how the i-IOSS gains result.
Firstly, we observe that the decrease function ${\alpha_{3}}$ increases as ${\| A_{L} \|_{P}}$ decreases.
Since, without loss of generality, ${Q}$ and ${P}$ can be uniformly rescaled without touching ${L}$ or ${A_{L}}$, decreasing ${\| A_{L} \|_{P}}$ essentially means decreasing the eigenvalues of ${A_{L}}$.
In order to drive the eigenvalues of ${A_{L}}$ to zero, larger matrices ${L}$ are needed.
Hence, the faster the decrease rate shall be rendered the larger get the resulting gains for the outputs~\eqref{eq:thm:LINiIOSSLyap:Ygain} and the output disturbance~\eqref{eq:thm:LINiIOSSLyap:Vgain}.
Note that this is in perfect accordance with the usual trade-off in Luenberger observer design where faster observer dynamics result in larger output noise gains.
Secondly, we observe that on the one hand the lower bound for the eigenvalues of ${A_{L}}$ is given by the largest eigenvalue of the non-observable modes.
On the other hand, even in the observable case in which ${A_{L}}$ can be rendered nilpotent, ${\| A_{L} \|_{P}}$ will in general be unequal to zero.
Hence, even for observable systems the functions ${\beta(r, t)}$ and ${\beta^{\Sigma}(r, t)}$ in~\eqref{eq:defn:DiscountingPredictoriIOSS} and \eqref{eq:thm:iIOSSsummable} respectively, will in general not vanish\footnote{Note that this observation is in contrast to the statement in~\cite[Remark~14]{Ji_et_al_MHE}.} (especially for ${1 < t < n_{x}}$), which becomes especially evident by~\eqref{eq:cor:LINiIOSSDirectSum:errortrajec} in the proof below.
Finally, the above proof shows how \emph{explicit} decrease rates and gains for the i-IOSS estimate can be derived for arbitrary systems.
For this, we emphasize that the inequalities~\eqref{eq:thm:LINiIOSSLyap:contraction} and~\eqref{eq:thm:LINiIOSSLyap:Lyap} are expected to be rather tight and so are the derived decrease rate and the gains~\eqref{eq:thm:LINiIOSSLyap:Wgain}-\eqref{eq:thm:LINiIOSSLyap:Ygain}.
However, slightly tighter bounds can be obtained by deriving an explicit expression for the difference trajectory ${x_{\Delta}(t) = \bar{x}(t) - \bar{\chi}(t)}$ via an induction of~\eqref{eq:thm:LINiIOSSLyap:differenceinteration}, i.e., by circumventing the use of a Lyapunov-function, as shown in the following result.

\begin{cor}
\label{cor:LINiIOSSDirectSum}
A linear system is detectable if and only if it is time-discounted i-IOSS according to Definition~\ref{defn:DiscountingPredictoriIOSS} with~\eqref{eq:defn:DiscountingPredictoriIOSS} replaced by~\eqref{eq:thm:iIOSSsummable}.
\end{cor}

\begin{proof}
Following the notation and the arguments in the proof of Theorem~\ref{thm:LINiIOSSLyap}, an induction of~\eqref{eq:thm:LINiIOSSLyap:differenceinteration} results in
\begin{align}
  x_{\Delta}(t) &= A_{L}^{t} x_{\Delta}(0) + \sum _{\tau = 1} ^{t} A_{L}^{\tau-1} [ E w_{\Delta}(t-\tau)
  \label{eq:cor:LINiIOSSDirectSum:errortrajec}
  \\
  & \hspace{-0.5cm}
  \notag
  + L F v_{\Delta}(t-\tau) + \BL u_{\Delta}(t-\tau) - L y_{\Delta}(t-\tau) ]
\end{align}
such that applying ${\| \cdot \|_{P}}$ to both sides, using the triangle-inequality and submultiplicativity gives the desired estimate~\eqref{eq:thm:iIOSSsummable} with
\begin{align}
  \label{eq:cor:LINiIOSSDirectSum:alpha1}
  \alpha_1(r) &:= \sqrt{\lambda_{\text{min}}( P )} \, r
  \\
  \label{eq:cor:LINiIOSSDirectSum:Xgain}
  \beta^{\Sigma}(r, t) &:=
  \| A_{L}^{t} P^{\frac{1}{2}} \|_{P} \, r
  \\
  \label{eq:cor:LINiIOSSDirectSum:Wgain}
  \gamma^{\Sigma}(r, t) &:=
  \| A_{L}^{t-1} E P^{\frac{1}{2}} \|_{P} \, r
  \\
  \label{eq:cor:LINiIOSSDirectSum:Vgain}
  \delta^{\Sigma}(r, t) &:=
  \| A_{L}^{t-1} L F P^{\frac{1}{2}} \|_{P} \, r
  \\
  \label{eq:cor:LINiIOSSDirectSum:Ugain}
  \epsilon^{\Sigma}(r, t) &:=
  \| A_{L}^{t-1} \BL P^{\frac{1}{2}} \|_{P} \, r
  \\
  \label{eq:cor:LINiIOSSDirectSum:Ygain}
  \varphi^{\Sigma}(r, t) &:=
  \| A_{L}^{t-1} L P^{\frac{1}{2}} \|_{P} \, r.
\end{align}
It remains to show that the above functions are of the desired function class (or can be bounded from above accordingly and arbitrarily close).
We observe
(i) that ${\alpha_1 \in \mathcal{K}}$ applies,
(ii) that the functions defined in~\eqref{eq:cor:LINiIOSSDirectSum:Xgain}-\eqref{eq:cor:LINiIOSSDirectSum:Ygain} are continuous,
(iii) that, for each ${t \geq 1}$, they are either zero or strictly increasing with respect to their first argument,
and (iv) that, for increasing ${t \geq 1}$, all ${\| \cdot \|_{P}}$-terms are non-increasing and converge to zero because ${\| A_{L} \|_{P} \in [0, 1)}$ is guaranteed.
(Note that the case ${t = 0}$ is critical only for irregular matrices ${A_{L}}$ and only for the functions defined in~\eqref{eq:cor:LINiIOSSDirectSum:Wgain}-\eqref{eq:cor:LINiIOSSDirectSum:Ygain}, which are never evaluated at ${t = 0}$ in the i-IOSS context.)
Hence, all definitions~\eqref{eq:cor:LINiIOSSDirectSum:Xgain}-\eqref{eq:cor:LINiIOSSDirectSum:Ygain} satisfy the conditions of ${\mathcal{KL}}$-functions (or take the value zero and can be bounded from above by arbitrary small ${\mathcal{KL}}$-functions).
\end{proof}

Finally, re-considering the motivation for introducing distinct terms for ${w}$, ${v}$, ${u}$, and ${y}$ in Definition~\ref{defn:DiscountingPredictoriIOSS}, the proofs of Theorem~\ref{thm:LINiIOSSLyap} and Corollary~\ref{cor:LINiIOSSDirectSum} illustrate how all four terms naturally result for linear systems.
Moreover, the definitions~\eqref{eq:thm:LINiIOSSLyap:Wgain}-\eqref{eq:thm:LINiIOSSLyap:Ygain} and~\eqref{eq:cor:LINiIOSSDirectSum:Wgain}-\eqref{eq:cor:LINiIOSSDirectSum:Ygain} show that for additive disturbances (${E = B}$ and ${F = I}$) and without input feed-through term ${(D = 0)}$, the process disturbance term, i.e., \eqref{eq:cor:LINiIOSSDirectSum:Wgain} or \eqref{eq:thm:LINiIOSSLyap:Wgain}, equals the input term, i.e., \eqref{eq:cor:LINiIOSSDirectSum:Ugain} respectively \eqref{eq:thm:LINiIOSSLyap:Ugain}, and the output disturbance term, i.e., \eqref{eq:cor:LINiIOSSDirectSum:Vgain} or \eqref{eq:thm:LINiIOSSLyap:Vgain}, equals the output term, i.e., \eqref{eq:cor:LINiIOSSDirectSum:Ygain} respectively \eqref{eq:thm:LINiIOSSLyap:Ygain}.

\section{CONCLUSIONS}
\label{sec:conclusions}

The present work provides two time-discounted i-IOSS formulations as a detectability notion for general nonlinear systems with non-additive disturbances.
Our definition covers previous i-IOSS notions for nonlinear systems as special cases as well as the linear case.
Furthermore, we prove that time-discounted i-IOSS can be shown by Lyapunov function techniques and that this property is necessary for the existence of RGAS observers.
The Lyapunov function techniques allow to verify i-IOSS in order to apply recent MHE results in the nonlinear context.
For their application in the linear case, explicit i-IOSS bounds are presented and discussed.

\end{document}